\documentclass[letterpaper,11pt]{article}

\usepackage{amsmath}
\usepackage{amsthm}
\usepackage{authblk}
\usepackage{setspace}
\usepackage{url}
\usepackage[comma,authoryear]{natbib}
\usepackage{hyperref}
\usepackage{theoremref}
\usepackage[titletoc,title]{appendix}
\usepackage[table]{xcolor}
\usepackage{float}
\restylefloat{table}
\usepackage{amssymb}
\usepackage{caption}
\usepackage{subcaption}
\usepackage[shortlabels]{enumitem}
\setcitestyle{citesep={;}}

\definecolor{armygreen}{rgb}{0.29, 0.33, 0.13}
\definecolor{auburn}{rgb}{0.43, 0.21, 0.1}
\definecolor{burgundy}{rgb}{0.5, 0.0, 0.13}
\definecolor{medium red}{rgb}{.490,.298,.337}
\definecolor{dark red}{rgb}{.235,.141,.161}
\definecolor{dark green}{rgb}{0.0,0.5,0.0}

\hypersetup{
	colorlinks = true,
	linkcolor = {burgundy},
	urlcolor = {burgundy},
	citecolor = {burgundy}, 		
	linkbordercolor = {white},
}

\newtheorem{theorem}{Theorem}
\newtheorem*{theorem*}{Theorem}
\newtheorem{proposition}{Proposition}
\newtheorem{claim}{Claim}[section]
\newtheorem{lemma}{Lemma}[section]
\newtheorem{corollary}{Corollary}

\theoremstyle{definition}
\newtheorem{definition}{Definition}
\theoremstyle{definition}
\newtheorem{example}{Example}
\theoremstyle{definition}
\newtheorem{remark}{Remark}
\theoremstyle{definition}
\newtheorem{note}{Note}

\newcommand*{\claimproofname}{Proof of Claim}
\newenvironment{claimproof}[1][\claimproofname]{\begin{proof}[#1]}{\end{proof}}

\usepackage{verbatim}
\usepackage{tikz}
\tikzset{
	treenode/.style = {shape=rectangle, rounded corners,
		draw, align=center,
		top color=white, bottom color=blue!20},
	root/.style = {shape=rectangle, rounded corners,
		draw, align=center,
		top color=white, bottom color=blue!20},
	root1/.style = {shape=rectangle, rounded corners,
		draw, align=center,
		top color=white, bottom color=red!40},
	root2/.style = {shape=rectangle, rounded corners,
		draw, align=center,
		top color=white, bottom color=red!20},
	env/.style = {shape=rectangle, rounded corners,
		draw, align=center,
		top color=white, bottom color=blue!20},
}

\title{Equivalence between individual and group strategy-proofness under stability\thanks{I thank Hector Chade and Alexander Teytelboym for helpful discussions.}}

\author{Pinaki Mandal\thanks{E-mail: \textit{pinaki.mandal@asu.edu}} }

\affil{Department of Economics, Arizona State University, USA}

\date{ }

\begin{document} 	
	\maketitle
	
	\begin{abstract}
		This paper studies the (group) strategy-proofness aspect of two-sided matching markets under stability. 
		For a one-to-one matching market, we show an equivalence between individual and group strategy-proofness under stability.
		We obtain this equivalence assuming the domain satisfies a richness condition. 
		However, the result cannot be extended to the many-to-one matching markets.
		We further consider a setting with single-peaked preferences and characterize all domains compatible for stability and (group) strategy-proofness.
	\end{abstract}

	\noindent \textbf{Keywords:} {Two-sided matching; Stability; Strategy-proofness; Group strategy-proofness; Single-peaked preferences} \\\medskip
	\noindent \textbf{JEL Classification:} C78; D71; D82

	\newpage

	\section{Introduction}

	The theory of two-sided matching markets has interested researchers for its relevance to the design of real-world institutions, such as assigning graduates to residency programs (National Resident Matching Program) or students to schools (Boston Public Schools). In this paper, we deal with the simplest case -- the \textit{marriage problem} \citep{gale1962college}, a well-known one-to-one matching market. In this market, there are two finite disjoint sets of agents, ``men'' and ``women''. Each agent on one side of the market has a strict preference over the agents on the other side and the \textit{outside option}, where the outside option denotes the possibility of remaining unmatched. A matching between men and women is selected based on the agents' preferences, where each agent on one side of the market can be matched with at most one agent on the other side.
	
	The \textit{deferred acceptance (DA) rule} \citep{gale1962college} is the salient rule for such a market due to its theoretical appeal.
	\begin{enumerate}[(a)]
		\item It is a \textit{stable} matching rule (see \citet{gale1962college}).\footnote{In real-world applications, empirical studies have shown that stable mechanisms often succeed whereas unstable ones often fail. For a summary of this evidence, see \citet{roth2002economist}.} A matching is stable if it is \textit{individually rational} and no pair of agents, one on each side, would rather be matched to each other than to their present match.
		
		\item For the proposing side, not only it is \textit{strategy-proof} but also \textit{group strategy-proof} (see \citet{dubins1981machiavelli}). A matching rule is strategy-proof if truthful revelation of preferences is a weakly dominant strategy for the agents. Group strategy-proofness, a stronger condition than strategy-proofness, ensures that no group of coordinated agents can be strictly better off by misreporting their preferences.
	\end{enumerate} 
	However, the DA rule is not strategy-proof for all agents (in fact, no stable matching rule is; see \citet{roth1982economics}), and consequently, is not group strategy-proof for all agents. 
	
	Our motivation behind this paper is twofold. First, analyze the structure of manipulative coalitions for the DA rule, and second, identify the conditions for the DA rule to be group strategy-proof. We are focusing on group strategy-proofness, not just on strategy-proofness; what use would it be to guarantee that no single agent could cheat if a few of them could jointly manipulate?
	
	As we have mentioned earlier, \citet{dubins1981machiavelli} show that no coalition of men can manipulate the \textit{men-proposing DA (MPDA) rule}, while \citet{roth1982economics} shows that the MPDA rule is not even strategy-proof for women. Therefore, whenever the MPDA rule is manipulable by a coalition, there must be at least one woman in that coalition. In Theorem \ref{theorem structure of manipulative coalitions}, we show that the manipulative coalition not only contains at least one woman but must be a group of women, extending the result of \citet{dubins1981machiavelli}. We further show that whenever a coalition manipulates the MPDA rule, every woman in the market weakly benefits while every man in the market weakly suffers (Proposition \ref{proposition update in welfare}), and the set of unmatched agents remains the same (Proposition \ref{proposition same set of matched agents}). One key implication of Proposition \ref{proposition same set of matched agents} is that an unmatched agent cannot be a part of manipulating the DA rule. 
	
	We next identify the conditions for the DA rule to be group strategy-proof. \citet{alcalde1994top} identify a restriction on the domain, called \textit{top dominance}, and show that top dominance for women is a sufficient condition for the MPDA rule to be strategy-proof.\footnote{Top dominance for women is also a necessary domain restriction for the MPDA rule to be strategy-proof under two domain conditions (see Theorem 4 in \citet{alcalde1994top}).} In Proposition \ref{proposition wgsp existence}, we show that it is also sufficient for the MPDA rule to be group strategy-proof. 
	As it turns out, this coincidence is not an implication of top dominance but rather a property of the DA rule. In particular, we find an equivalence between strategy-proofness and group strategy-proofness for the DA rule. 
	We obtain this equivalence assuming the domain satisfies a richness condition, called \textit{unrestricted top pairs} \citep{alva2017manipulation}, for at least one side of the market. The richness condition roughly requires that for every ordered pair of outcomes for an agent, there is an admissible preference that ranks them first and second. For example, if every strict preference is admissible for every man, the corresponding domain satisfies unrestricted top pairs for men. 
	Another result of ours shows that on a domain satisfying unrestricted top pairs for at least one side of the market, if there exists a stable and strategy-proof matching rule, it must be unique and is the DA rule (Lemma \ref{lemma only if}). 
	Combining all these results, we have our main result -- the equivalence between strategy-proofness and group strategy-proofness under stability (Theorem \ref{theorem sp wgsp equivalence}). 
	\citet{barbera2016group} show such an equivalence for general private good economies (which also encompasses the marriage problem). They obtain their result assuming a richness condition of the domain and two conditions of the rule (neither of them is stability). However, there is no connection between our richness condition and theirs (i.e., neither of them implies the other), and they assume (group) strategy-proofness only for one side of the market. Therefore, our equivalence result cannot be deduced from their result.
	
	So far, all the results assume that the agents can have strict but otherwise arbitrary preferences. However, in many circumstances, preferences may well be restricted. A natural restriction is as follows: the agents on each side are ordered based on a certain criterion, say age, and the preferences respect these orders in the sense that as one moves away from his/her most preferred choice, his/her preference declines. Such restriction is known as \textit{single-peakedness} \citep{black1948rationale}. In Theorem \ref{theorem single-peaked}, we characterize all single-peaked domains compatible for stability and (group) strategy-proofness under two domain conditions; namely, \textit{cyclical inclusion} and \textit{anonymity}. 
	Cyclical inclusion roughly requires that for every pair of outcomes for an agent, if there exists an admissible preference that prefers the first outcome to the second, then there is another admissible preference that prefers the second outcome to the first. 
	Anonymity requires every man to have the same set of admissible preferences, and so for the women. 
	Unlike Theorem \ref{theorem sp wgsp equivalence}, Theorem \ref{theorem single-peaked} does not show equivalence between strategy-proofness and group strategy-proofness under stability. It rather shows that under cyclical inclusion and anonymity, the set of single-peaked domains compatible for stability and strategy-proofness is the same as that for stability and group strategy-proofness.
	
	Finally, in Section \ref{section many-to-one}, we discuss to which extent we can generalize our results to many-to-one matching markets. We briefly describe a well-studied market -- the \textit{college admissions problem} \citep{gale1962college}, and show the impossibility of extending our results to this market by means of an example.
	
	The layout of the paper is as follows. In Section \ref{section preliminaries}, we introduce basic notions and notations that we use throughout the paper, describe our model, define matching rules and discuss their standard properties, and present the DA rule. Section \ref{section results} presents our results. Section \ref{section single-peaked} considers a setting with single-peaked preferences. In Section \ref{section many-to-one}, we discuss to which extent we can generalize our results to many-to-one matching markets. Finally, the Appendix contains the proofs.

	\section{Preliminaries}\label{section preliminaries}

	\subsection{Basic notions and notations}

	For a finite set $X$, let $\mathbb{L}(X)$ denote the set of all strict linear orders over $X$.\footnote{A \textbf{\textit{strict linear order}} is a semiconnex, asymmetric, and transitive binary relation.} An element of $\mathbb{L}(X)$ is called a \textbf{\textit{preference}} over $X$. 
	For a preference $P \in \mathbb{L}(X)$ and distinct $x, y \in X$, $x \mathrel{P} y$ is interpreted as ``$x$ is preferred to $y$ according to $P$''.
	For $P \in \mathbb{L}(X)$, let $R$ denote the weak part of $P$, i.e., for any $x, y \in X$, $x \mathrel{R} y$ if and only if \big[$x \mathrel{P} y$ or $x = y$\big].
	Furthermore, for $P \in \mathbb{L}(X)$ and non-empty $X' \subseteq X$, let $\tau(P, X')$ denote the most preferred element in $X'$ according to $P$, i.e., $\tau(P,X') = x$ if and only if \big[$x \in X'$ and $x \mathrel{P} y$ for all $y \in X' \setminus \{x\}$\big]. For ease of presentation, we denote $\tau(P, X)$ by $\tau(P)$.

	\subsection{Model}

	There are two finite disjoint sets of agents, the set of \textit{men} $M = \{m_1, \ldots, m_p\}$ and the set of \textit{women} $W = \{w_1, \ldots, w_q\}$. Let $A = M \cup W$ be the set of all agents. Throughout this paper, we assume $p, q \geq 2$. Let $\emptyset$ denote the \textbf{\textit{outside option}} -- the null agent.	
	
	Each man $m$ has a preference $P_m$ over $W \cup \{\emptyset\}$, the set of all women and the outside option. The position in which he places the outside option in the preference has the meaning that the only women he is willing to be matched with are those whom he prefers to the outside option. Similarly, each woman $w$ has a preference $P_w$ over $M \cup \{\emptyset\}$. We say that woman $w$ is \textbf{\textit{acceptable}} to man $m$ if $w \mathrel{P_m} \emptyset$, and analogously, man $m$ is \textit{acceptable} to woman $w$ if $m \mathrel{P_w} \emptyset$.
	
	We denote by $\mathcal{P}_a$ the set of admissible preferences for agent $a \in A$. Clearly, $\mathcal{P}_m \subseteq \mathbb{L}(W \cup \{\emptyset\})$ for all $m \in M$ and $\mathcal{P}_w \subseteq \mathbb{L}(M \cup \{\emptyset\})$ for all $w \in W$.
	A \textbf{\textit{preference profile}}, denoted by $P_A = (P_{m_1}, \ldots, P_{m_p}, P_{w_1}, \ldots, P_{w_q})$, is an element of the Cartesian product $\mathcal{P}_A := \underset{i=1}{\overset{p}{\prod}}\mathcal{P}_{m_i} \times \underset{j=1}{\overset{q}{\prod}}\mathcal{P}_{w_j}$, that represents a collection of preferences -- one for each agent.
	Furthermore, as is the convention, $P_{-a}$ denotes a collection of preferences of all agents except for $a$. Also, for $A' \subseteq A$, let $P_{A'}$ denote a collection of preferences of all agents in $A'$ and $P_{-A'}$ a collection of preferences of all agents not in $A'$.

	\subsection{Matching rules and their stability}

	A \textbf{\textit{matching}} (between $M$ and $W$) is a function $\mu : A \to A \cup \{\emptyset\}$ such that 
	\begin{enumerate}[(i)]
		\item\label{item matching definition 1} $\mu(m) \in W \cup \{\emptyset\}$ for all $m \in M$,
		
		\item\label{item matching definition 2} $\mu(w) \in M \cup \{\emptyset\}$ for all $w \in W$, and
		
		\item\label{item matching definition 3} $\mu(m) = w$ if and only if $\mu(w) = m$ for all $m \in M$ and all $w \in W$.
	\end{enumerate}
	Here, $\mu(m) = w$ means man $m$ and woman $w$ are matched to each other under the matching $\mu$, and $\mu(a) = \emptyset$ means agent $a$ is unmatched under the matching $\mu$. We denote by $\mathcal{M}$ the set of all matchings.
	
	A matching $\mu$ is \textbf{\textit{individually rational}} at a preference profile $P_A$ if for every $a \in A$, we have $\mu(a) \mathrel{R_a} \emptyset$. A matching $\mu$ is \textit{blocked} by a pair $(m,w) \in M \times W$ at a preference profile $P_A$ if $w \mathrel{P_m} \mu(m)$ and $m \mathrel{P_w} \mu(w)$. A matching is \textbf{\textit{stable}} at a preference profile if it is individually rational and is not blocked by any pair at that preference profile.
	
	A \textbf{\textit{matching rule}} is a function $\varphi: \mathcal{P}_A \to \mathcal{M}$. For a matching rule $\varphi: \mathcal{P}_A \to \mathcal{M}$ and a preference profile $P_A \in \mathcal{P}_A$, let $\varphi_a(P_A)$ denote the match of agent $a$ by $\varphi$ at $P_A$.

	\begin{definition}\label{definition stability}
		A matching rule $\varphi: \mathcal{P}_A \to \mathcal{M}$ is \textbf{\textit{stable}} if for every $P_A \in \mathcal{P}_A$, $\varphi(P_A)$ is stable at $P_A$.
	\end{definition}

	\subsection{Incentive properties of matching rules}

	In practice, matching rules are often designed to satisfy incentive properties. Two well-studied such requirements are \textit{strategy-proofness} and \textit{group strategy-proofness}.

	\begin{definition}\label{definition incentive properties}
		A matching rule $\varphi: \mathcal{P}_A \to \mathcal{M}$ is
		\begin{enumerate}[(i)]
			\item \textbf{\textit{strategy-proof}} if for every $P_A \in \mathcal{P}_A$, every $a \in A$, and every $\tilde{P}_a \in \mathcal{P}_a$, we have $\varphi_a(P_A) \mathrel{R_a} \varphi_a(\tilde{P}_a, P_{-a})$.
			
			\item \textit{\textbf{group strategy-proof}} if for every $P_A \in \mathcal{P}_A$, there do not exist a set of agents $A' \subseteq A$ and a preference profile $\tilde{P}_{A'}$ of the agents in $A'$ such that $\varphi_a(\tilde{P}_{A'},P_{-A'}) \mathrel{P_a} \varphi_a(P_A)$ for all $a \in A'$.
		\end{enumerate}
	\end{definition}

	If a matching rule $\varphi$ on $\mathcal{P}_A$ is not group strategy-proof, then there exist $P_A \in \mathcal{P}_A$, a set of agents $A' \subseteq A$, and a preference profile $\tilde{P}_{A'}$ of the agents in $A'$ such that $\varphi_a(\tilde{P}_{A'},P_{-A'}) \mathrel{P_a} \varphi_a(P_A)$ for all $a \in A'$. In such cases, we say that \textit{$\varphi$ is manipulable at $P_A$ by coalition $A'$ via $\tilde{P}_{A'}$}. Note that a coalition can be a singleton, and thus, group strategy-proofness implies strategy-proofness. 
	
	Notice that all agents in the manipulative coalition should be strictly better off from misreporting. We consider this requirement compelling, since it leaves no doubt regarding the incentives for each member of the coalition to participate in a collective deviation from truthful revelation.

	\subsection{Deferred acceptance}\label{section DA}

	\textit{Deferred acceptance (DA) rule} \citep{gale1962college} is the salient rule in our model for its theoretical appeal. 
	\begin{enumerate}[(a)]
		\item It is a stable matching rule. In fact, it is the stable matching rule optimal for the agents on the ``proposing'' side (see \citet{gale1962college} for details).
		
		\item For the ``proposing'' side of the market, not only it is strategy-proof but also group strategy-proof (see \citet{dubins1981machiavelli} for details).
	\end{enumerate}	
	
	There are two types of the DA rule: the \textit{men-proposing DA (MPDA) rule} -- denoted by $D^M$, and the \textit{women-proposing DA (WPDA) rule}. In the following, we provide a description of the MPDA rule at a preference profile $P_A$. The same of the WPDA rule can be obtained by interchanging the roles of men and women in the MPDA rule.

	\begin{itemize}[leftmargin = 1.35cm]
		\item[\textbf{\textit{Step 1.}}] Each man $m$ proposes to his most preferred acceptable woman (according to $P_m$).\footnote{That is, if the most preferred woman of a man is acceptable to that man, he proposes to her. Otherwise, he does not propose to anybody.} Every woman $w$, who has at least one proposal, tentatively keeps her most preferred acceptable man (according to $P_w$) among these proposals and rejects the rest.
		
		\item[\textbf{\textit{Step $2$.}}] Every man $m$, who was rejected in the previous step, proposes to his next preferred acceptable woman. Every woman $w$, who has at least one proposal including any proposal tentatively kept from the earlier steps, tentatively keeps her most preferred acceptable man among these proposals and rejects the rest.
	\end{itemize}

	This procedure is then repeated from Step 2 till a step such that for each man, one of the following two happens: (i) he is accepted by some woman, (ii) he has proposed to all acceptable women. At this step, the proposal tentatively accepted by women becomes permanent. This completes the description of the MPDA rule.

	\begin{remark}[\citealp{gale1962college}]\label{remark gale result}
		On the unrestricted domain $\mathbb{L}^p(W \cup \{\emptyset\}) \times \mathbb{L}^q(M \cup \{\emptyset\})$, both the DA rules are stable.
	\end{remark}

	\section{Results}\label{section results}

	\subsection{Structure of manipulative coalitions for the MPDA rule}

	\citet{dubins1981machiavelli} show that no coalition of men can manipulate the MPDA rule on the unrestricted domain, while \citet{roth1982economics} shows that no stable matching rule on the unrestricted domain is strategy-proof.\footnote{\citet{roth1982economics} proves this result in a setting without outside options and with an equal number (at least three) of men and women. However, the result can be extended to our setting (i.e., with outside options and with arbitrary values (at least two) of the number of men and the number of women). See Example \ref{example both side unrestricted no sp rule} for a stronger result.} In view of these results, it follows that whenever the MPDA rule is manipulable by a coalition, at least one woman must be in that coalition. It turns out that the coalition not only contains at least one woman, but must be a group of women.

	\begin{theorem}\label{theorem structure of manipulative coalitions}
		Suppose a coalition $A' \subseteq A$ manipulates the MPDA rule at some preference profile. Then, $A' \subseteq W$.
	\end{theorem}

	The result of \citet{dubins1981machiavelli} follows from Theorem \ref{theorem structure of manipulative coalitions}.

	\begin{corollary}[\citealp{dubins1981machiavelli}]\label{corollary dubins}
		On the unrestricted domain $\mathbb{L}^p(W \cup \{\emptyset\}) \times \mathbb{L}^q(M \cup \{\emptyset\})$, no coalition of men can manipulate the MPDA rule.
	\end{corollary}

	Our next result concerns how a manipulation affects the MPDA rule from the agents' point of view. While every woman in the market weakly benefits from a successful misreporting, each man weakly suffers. Notice that Theorem \ref{theorem structure of manipulative coalitions} follows from this result.

	\begin{proposition}\label{proposition update in welfare}
		On an arbitrary domain $\mathcal{P}_A$, suppose the MPDA rule $D^M$ is manipulable at $P_A \in \mathcal{P}_A$ by coalition $A' \subseteq A$ via $\tilde{P}_{A'} \in \underset{a \in A'}{\prod}\mathcal{P}_{a}$. Then, 
		\begin{enumerate}[(a)]
			\item\label{item structure men weakly worse} $D^M_m(P_A) \mathrel{R_m} D^M_m(\tilde{P}_{A'},P_{-A'})$ for all $m \in M$, and
			
			\item\label{item structure women weakly better} $D^M_w(\tilde{P}_{A'},P_{-A'}) \mathrel{R_w} D^M_w(P_A)$ for all $w \in W$.
		\end{enumerate}
	\end{proposition}

	Our last result in this subsection says that the set of unmatched agents does not get affected by manipulation.

	\begin{proposition}\label{proposition same set of matched agents}
		On an arbitrary domain $\mathcal{P}_A$, suppose the MPDA rule $D^M$ is manipulable at $P_A \in \mathcal{P}_A$ by coalition $A' \subseteq A$ via $\tilde{P}_{A'} \in \underset{a \in A'}{\prod}\mathcal{P}_{a}$. Then, for every $a \in A$,
		\begin{equation*}
			D^M_a(P_A) = \emptyset \iff D^M_a(\tilde{P}_{A'},P_{-A'}) = \emptyset.\footnotemark
		\end{equation*}
	\end{proposition}
	\footnotetext{Throughout, ``$A \implies B$'' means that ``$A$ implies $B$'', and ``$A \iff B$'' means that ``$A$ if and only if $B$''.}

	A key implication of Proposition \ref{proposition same set of matched agents} is that an unmatched agent cannot be a part of a manipulation. Note that by symmetry, Proposition \ref{proposition same set of matched agents} also holds for the WPDA rule.
	
	Proposition \ref{proposition same set of matched agents} cannot be deduced from a result of \citet{mcvitie1970stable}, where they show that the set of unmatched agents remains the same across the stable matchings at a preference profile. To see this, note that in Proposition \ref{proposition same set of matched agents}, the matching $D^M(P_A)$ is not stable at $(\tilde{P}_{A'},P_{-A'})$, and the matching $D^M(\tilde{P}_{A'},P_{-A'})$ is not stable at $P_A$ in general.

	\subsection{Equivalence between strategy-proofness and group strategy-proofness}

	\citet{alcalde1994top} identify a restriction on the domain, called \textit{top dominance}, and show that top dominance for women is a sufficient condition for the MPDA rule to be strategy-proof. Our next result, which extends their result, says it is also sufficient for the MPDA rule to be group strategy-proof. Before stating this result, we first present the notion of top dominance.

	\begin{definition}[Top dominance]
		A domain of preference profiles $\mathcal{P}_A$ satisfies \textit{\textbf{top dominance} for women} if for every $w \in W$, $\mathcal{P}_w$ satisfies the following property: for every $x \in M$ and every $y, z \in M \cup \{\emptyset\}$, if there exists a preference $P \in \mathcal{P}_w$ with $x \mathrel{P} y \mathrel{P} z$ and $y \mathrel{R} \emptyset$, then there is no preference $\tilde{P} \in \mathcal{P}_w$ such that $x \mathrel{\tilde{P}} z \mathrel{\tilde{P}} y$ and $z \mathrel{\tilde{R}} \emptyset$.
	\end{definition}

	We define \textit{\textbf{top dominance} for men} in an analogous manner.

	\begin{proposition}\label{proposition wgsp existence}
		Let $\mathcal{P}_A$ be an arbitrary domain of preference profiles. If $\mathcal{P}_A$ satisfies top dominance for women, then the MPDA rule is stable and group strategy-proof on $\mathcal{P}_A$.
	\end{proposition}

	As we can see, for the MPDA rule, strategy-proofness coincides with group strategy-proofness when the domain satisfies top dominance for women. This coincidence is not an implication of top dominance, but rather a property of any stable matching rule. This is what we show next. We first introduce a richness condition, called \textit{unrestricted top pairs} \citep{alva2017manipulation}, of the domain that features in our result.

	\begin{definition}[Unrestricted top pairs]\label{definition unrestricted top pairs}
		A domain of preference profiles $\mathcal{P}_A$ satisfies \textit{\textbf{unrestricted top pairs} for men} if for every $m \in M$,
		\begin{enumerate}[(i)]
			\item for every $w, w' \in W$, there exists $P \in \mathcal{P}_m$ such that $w \mathrel{P} w' \mathrel{P} z$ for all $z \in (W \cup \{\emptyset\}) \setminus \{w,w'\}$,
			
			\item\label{item unrestricted top} for every $w \in W$, there exists $\tilde{P} \in \mathcal{P}_m$ such that $w \mathrel{\tilde{P}} \emptyset \mathrel{\tilde{P}} z$ for all $z \in W \setminus \{w\}$, and
			
			\item there exists $P' \in \mathcal{P}_m$ such that $\tau(P') = \emptyset$.
		\end{enumerate}
	\end{definition}

	For example, whenever the sets of admissible preferences for men are unrestricted, the corresponding domain satisfies unrestricted top pairs for men. We define \textit{\textbf{unrestricted top pairs} for women} in an analogous manner.
	
	We now present our main result of this paper. It shows the equivalence between strategy-proofness and group strategy-proofness for any stable matching rule under our richness condition.

	\begin{theorem}\label{theorem sp wgsp equivalence}
		Let $\mathcal{P}_A$ satisfy unrestricted top pairs for at least one side of the market. Then, any stable matching rule on $\mathcal{P}_A$ is strategy-proof if and only if it is group strategy-proof.
	\end{theorem}

	Note that the domain satisfying unrestricted top pairs for at least one side of the market is a sufficient condition for the equivalence between strategy-proofness and group strategy-proofness under stability, not for the existence of a stable and (group) strategy-proof matching rule. In fact, if the domain satisfies unrestricted top pairs for both sides, no stable matching rule is strategy-proof (see Example \ref{example both side unrestricted no sp rule}), and therefore, Theorem \ref{theorem sp wgsp equivalence} is vacuously satisfied.

	\begin{example}\label{example both side unrestricted no sp rule}
		Suppose the domain $\mathcal{P}_A$ satisfies unrestricted top pairs for both sides. Consider the preference profiles presented in Table \ref{preference profiles unrestricted for both}. For instance, $w_1 w_2 \ldots$ denotes a preference that ranks $w_1$ first and $w_2$ second (the dots indicate that all preferences for the corresponding parts are irrelevant and can be chosen arbitrarily). Here, $m_k$ denotes a man other than $m_1, m_2$ (if any), and $w_l$ denotes a woman other than $w_1, w_2$ (if any). Note that only man $m_1$ changes his preference from $P^1_A$ to $P^2_A$ and only woman $w_1$ changes her preference from $P^1_A$ to $P^3_A$.
		\begin{table}[H]
			\centering
			\begin{tabular}{@{}c|ccc|ccc@{}}
				\hline
				Preference profiles & $m_1$ & $m_2$ & $m_k$ & $w_1$ & $w_2$ & $w_l$ \\ \hline
				\hline
				$P^1_A$ & $w_1 w_2 \ldots$ & $w_2 w_1 \ldots$ & $\emptyset \ldots$ & $m_2 m_1 \ldots$ & $m_1 m_2 \ldots$ & $\emptyset \ldots$ \\
				$P^2_A$ & $w_1 \emptyset \ldots$ & $w_2 w_1 \ldots$ & $\emptyset \ldots$ & $m_2 m_1 \ldots$ & $m_1 m_2 \ldots$ & $\emptyset \ldots$ \\
				$P^3_A$ & $w_1 w_2 \ldots$ & $w_2 w_1 \ldots$ & $\emptyset \ldots$ & $m_2 \emptyset \ldots$ & $m_1 m_2 \ldots$ & $\emptyset \ldots$ \\
				\hline
			\end{tabular}
			\caption{Preference profiles for Example \ref{example both side unrestricted no sp rule}}
			\label{preference profiles unrestricted for both}
		\end{table}
		
		For ease of presentation, let $\mu$ denote the matching $\big[(m_1,w_1), (m_2,w_2), (a,\emptyset) \hspace{2 mm} \forall \hspace{2 mm} a \in A \setminus \{m_1, m_2, w_1, w_2\}\big]$ and $\tilde{\mu}$ the matching $\big[(m_1,w_2), (m_2,w_1), (a,\emptyset) \hspace{2 mm} \forall \hspace{2 mm} a \in A \setminus \{m_1, m_2, w_1, w_2\}\big]$ in this example. 
		The sets of stable matchings at $P^1_A$, $P^2_A$, and $P^1_A$ are $\{\mu, \tilde{\mu}\}$, $\{\mu\}$, and $\{\tilde{\mu}\}$, respectively. 
		
		Fix a stable matching rule $\varphi$ on $\mathcal{P}_A$. If $\varphi(P^1_A) = \mu$, then $w_1$ can manipulate at $P^1_A$ via $P^3_{w_1}$. If $\varphi(P^1_A) = \tilde{\mu}$, then $m_1$ can manipulate at $P^1_A$ via $P^2_{m_1}$. This implies $\varphi$ is not strategy-proof on $\mathcal{P}_A$.
		\hfill
		$\Diamond$
	\end{example}

	\section{Application: Results with single-peaked preferences}\label{section single-peaked}

	So far, all the results assume that the agents can have strict but otherwise arbitrary preferences. However, in many circumstances, preferences may well be restricted. A natural restriction is as follows: the agents on each side are ordered based on a certain criterion, say age, and the preferences respect these orders in the sense that as one moves away from his/her most preferred choice, his/her preference declines. Such restriction is known as \textit{single-peakedness} \citep{black1948rationale} in the literature. This section explores the compatibility between stability and (group) strategy-proofness with single-peaked preferences. 
	
	Let $\prec_M$ be a prior ordering over $M$ and $\prec_W$ a prior ordering over $W$. 
	A preference $P \in \mathbb{L}(W \cup \{\emptyset\})$ is \textit{\textbf{single-peaked} with respect to $\prec_W$} if for every $w, w' \in W$, $\big[\tau(P, W) \preceq_W w \prec_W w' \mbox{ or } w' \prec_W w \preceq_W \tau(P, W)\big]$ implies $w \mathrel{P} w'$.\footnote{$\preceq_W$ denotes the weak part of $\prec_W$.} Similarly, a preference $P \in \mathbb{L}(M \cup \{\emptyset\})$ is \textit{\textbf{single-peaked} with respect to $\prec_M$} if for every $m, m' \in M$, $\big[\tau(P, M) \preceq_M m \prec_M m' \mbox{ or } m' \prec_M m \preceq_M \tau(P, M)\big]$ implies $m \mathrel{P} m'$. 
	Let $\mathbb{S}(\prec_M)$ and $\mathbb{S}(\prec_W)$ denote the sets of all single-peaked preferences with respect to $\prec_M$ and $\prec_W$, respectively. Throughout this section, whenever we write $\mathcal{P}_A$ is a single-peaked domain, we mean $\mathcal{P}_m \subseteq \mathbb{S}(\prec_W)$ for all $m \in M$ and $\mathcal{P}_w \subseteq \mathbb{S}(\prec_M)$ for all $w \in W$.
	
	Our first result in this section is an impossibility result, which shows the incompatibility between stability and strategy-proofness on the ``maximal'' single-peaked domain. To see this, notice that whenever there are only two men and two women in the market, the maximal single-peaked domain is equivalent to the unrestricted domain. Because of this, and since stability and strategy-proofness are incompatible on the unrestricted domain (see \citet{roth1982economics}), our result follows (using the logic of embedding the smaller market to a larger market in a trivial manner).

	\begin{proposition}\label{proposition no stable and strategy-proof on full single-peaked}
		On the maximal single-peaked domain $\mathbb{S}^p(\prec_W) \times \mathbb{S}^q(\prec_M)$, no stable matching rule is strategy-proof.
	\end{proposition}

	We next focus on identifying all single-peaked domains with at least one stable and (group) strategy-proof matching rule. We do so under two domain conditions; namely, \textit{cyclical inclusion} -- a richness condition, and \textit{anonymity}. First, we introduce these conditions. We use the following convention to do so: for a finite set $X$ and a set of preferences $\mathcal{D} \subseteq \mathbb{L}(X)$, whenever we write $(\cdot x \cdot y \cdot z \cdot) \in \mathcal{D}$, we mean that there exists $P \in \mathcal{D}$ such that $x \mathrel{P} y \mathrel{P} z$, where $x, y, z \in X$. Here, $x$ and $y$ need not be consecutively ranked in $P$, and neither are $y$ and $z$. Furthermore, $x$ need not be the most preferred element according to $P$ and $z$ need not be the least preferred element.

	\begin{definition}[Cyclical inclusion]\label{definition cyclical inclusion}
		A domain of preference profiles $\mathcal{P}_A$ satisfies \textit{\textbf{cyclical inclusion} for men} if for every $m \in M$,
		\begin{enumerate}[(i)]
			\item there exists $P \in \mathcal{P}_m$ such that $\tau(P) = \emptyset$, and
			
			\item for every $w, w' \in W$, $(\cdot w \cdot w' \cdot \emptyset \cdot) \in \mathcal{P}_m$ implies $(\cdot w' \cdot w \cdot \emptyset \cdot) \in \mathcal{P}_m$.
		\end{enumerate}
	\end{definition}

	In other words, cyclical inclusion for men requires that every man has an admissible preference with no acceptable women. It also requires that for any two women, say $w$ and $w'$, whenever a man, say $m$, has an admissible preference $P_m$ such that $w \mathrel{P_m} w' \mathrel{P_m} \emptyset$, then $m$ has another admissible preference $\tilde{P}_m$ with $w' \mathrel{\tilde{P}_m} w \mathrel{\tilde{P}_m} \emptyset$. We define \textit{\textbf{cyclical inclusion} for women} in an analogous manner.

	\begin{definition}[Anonymity]\label{definition anonimity}
		A domain of preference profiles $\mathcal{P}_A$ is \textit{\textbf{anonymous}} if
		\begin{enumerate}[(i)]
			\item $\mathcal{P}_m = \mathcal{P}_{m'}$ for all $m, m' \in M$, and
			
			\item $\mathcal{P}_w = \mathcal{P}_{w'}$ for every $w, w' \in W$.
		\end{enumerate} 
	\end{definition}

	In other words, anonymity requires every man to have the same set of admissible preferences, and so do the women.
	
	Our final result characterizes all single-peaked domains compatible for stability and (group) strategy-proofness under cyclical inclusion and anonymity. It further says that whenever such a matching rule exists, at least one of the DA rules is stable and group strategy-proof.

	\begin{theorem}\label{theorem single-peaked}
		Let $\mathcal{P}_A$ be an anonymous single-peaked domain that satisfies cyclical inclusion for both sides. The following statements are equivalent:
		\begin{enumerate}[(a)]
			\item\label{item TD} $\mathcal{P}_A$ satisfies top dominance for at least one side of the market.
			
			\item\label{item stable and sp} There exists a stable and strategy-proof matching rule on $\mathcal{P}_A$.
			
			\item\label{item stable and gsp} There exists a stable and group strategy-proof matching rule on $\mathcal{P}_A$.
			
			\item\label{item DA is gsp} At least one of the DA rules satisfies stability and strategy-proofness.
		\end{enumerate}
	\end{theorem}

	Notice that unlike Theorem \ref{theorem sp wgsp equivalence}, Theorem \ref{theorem single-peaked} does not show equivalence between strategy-proofness and group strategy-proofness for any stable matching rule. It rather shows that under cyclical inclusion and anonymity, the set of single-peaked domains compatible for stability and strategy-proofness is the same as that for stability and group strategy-proofness.

	\section{Discussion: Many-to-one matching markets}\label{section many-to-one}

	In this section, we briefly describe the \textit{college admissions problem} \citep{gale1962college}, a well-known many-to-one matching market, and discuss the impossibility of extending our results to such a matching market.

	\subsection{Model}

	There are two finite disjoint sets of agents, the set of \textit{colleges} $C$ and the set of \textit{students} $S$. Let $I = C \cup S$ be the set of all agents. 
	Each college $c$ has a \textit{quota} $q_c \geq 1$ which represents the maximum number of students for which it has places. 
	Let $\mathcal{S}_q := \{\tilde{S} \subseteq S \mid | \tilde{S} | \leq q \}$ be the set of subsets of $S$ with cardinality at most $q$. 
	Each college $c$ has a preference $P_c$ over $\mathcal{S}_{q_c}$ and each student $s$ has a preference $P_s$ over $C \cup \{\emptyset\}$. 
	We say that student $s$ is \textbf{\textit{acceptable}} to college $c$ if $\{s\} \mathrel{P_c} \emptyset$, and analogously, college $c$ is \textit{acceptable} to student $s$ if $c \mathrel{P_s} \emptyset$. 
	We denote by $\mathcal{P}_i$ the set of admissible preferences for agent $i \in I$.
	A \textbf{\textit{preference profile}}, denoted by $P_I = \big((P_c)_{c \in C}, (P_s)_{s \in S} \big)$, is an element of the Cartesian product $\mathcal{P}_I := \underset{c \in C}{\prod}\mathcal{P}_c \times \underset{s \in S}{\prod}\mathcal{P}_s$.
	
	We impose a standard assumption on the preferences of colleges, called \textit{responsiveness} \citep{roth1985college}. This property demonstrates a natural way to extend the preferences of colleges from individual students to sets of students.

	\begin{definition}[Responsiveness]
		A college $c$'s preference $P_c$ satisfies \textit{\textbf{responsiveness}} if for every $\tilde{S} \subseteq S$ with $|\tilde{S}| < q_c$,
		\begin{enumerate}[(i)]
			\item for every $s \in S \setminus \tilde{S}$,
			\begin{equation*}
				(\tilde{S} \cup \{s\}) \mathrel{P_c} \tilde{S} \iff \{s\} \mathrel{P_c} \emptyset, \mbox{ and}
			\end{equation*}
			
			\item for every $s, s' \in S \setminus \tilde{S}$,
			\begin{equation*}
				(\tilde{S} \cup \{s\}) \mathrel{P_c} (\tilde{S} \cup \{s'\}) \iff \{s\} \mathrel{P_c} \{s'\}.
			\end{equation*}
		\end{enumerate}
	\end{definition}

	A (many-to-one) \textbf{\textit{matching}} (between $C$ and $S$) is a function $\nu: I \rightarrow 2^{S} \cup C$ such that
	\begin{enumerate}[(i)]
		\item $\nu(c) \in \mathcal{S}_{q_c}$ for all $c \in C$,
		
		\item $\nu(s) \in C \cup \{\emptyset\}$ for all $s \in S$, and
		
		\item $\nu(s) = c$ if and only if $s \in \nu(c)$ for all $c \in C$ and all $s \in S$.
	\end{enumerate}
	
	A matching $\nu$ is \textit{blocked by a college $c$} at a preference profile $P_I$ if there exists $s \in \nu(c)$ such that $\emptyset \mathrel{P_c} \{s\}$. 
	A matching $\nu$ is \textit{blocked by a student $s$} at a preference profile $P_I$ if $\emptyset \mathrel{P_s} \nu(s)$. A matching $\nu$ is \textit{\textbf{individually rational}} at a preference profile $P_I$ if it is not blocked by any college or student. A matching $\nu$ is \textit{blocked by a pair $(c,s) \in C \times S$} at a preference profile $P_I$ if $c \mathrel{P_s} \nu(s)$ and either (i) $\big[|\nu(c)| < q_c$ and $\{s\} \mathrel{P_c} \emptyset \big]$, or (ii) \big[there exists $s' \in \nu(c)$ such that $\{s\} \mathrel{P_c} \{s'\}\big]$.
	A matching $\nu$ is \textit{\textbf{stable}} at a preference profile $P_I$ if it is individually rational and is not blocked by any pair at that preference profile.
	
	The DA rule naturally extends to the college admissions problem. In the following, we provide a description of the \textit{students-proposing DA (SPDA) rule}.

	\begin{itemize}[leftmargin = 1.35cm]
		\item[\textbf{\textit{Step 1.}}] Each student $s$ applies to her most preferred acceptable college. Every college $c$, which has at least one application, tentatively keeps its top $q_c$ acceptable students among these applications and rejects the rest (if $c$ has fewer acceptable applications than $q_c$, it tentatively keeps all of them).
		
		\item[\textbf{\textit{Step $2$.}}] Every student $s$, who was rejected in the previous step, applies to her next preferred acceptable college. Every college $c$, which has at least one application including any applications tentatively kept from the earlier steps, tentatively keeps its top $q_c$ acceptable students among these applications and rejects the rest.
	\end{itemize}

	This procedure is then repeated from Step 2 till a step such that for each student, one of the following two happens: (i) she is accepted by some college, (ii) she has applied to all acceptable colleges. At this step, the applications tentatively accepted by colleges become permanent. This completes the description of the SPDA rule.

	\subsection{Impossibility of extensions}

	Colleges having responsive preferences is a sufficient condition for the SPDA rule to be strategy-proof for students \citep{roth1985college}.\footnote{\citet{roth1985college} also shows that colleges having responsive preferences is not a sufficient condition for the \textit{colleges-proposing DA rule} to be strategy-proof for colleges.} Not only that, even a coalition of students cannot manipulate the SPDA rule when colleges have responsive preferences.\footnote{\citet{martinez2004group} extend this result further. They show that no coalition of students can manipulate the SPDA rule when the colleges' preferences satisfy both \textit{substitutability} \citep{kelso1982job} and \textit{separability} \citep{barbera1991voting}.} 
	But what if a group of students collude with a non-empty group (possibly singleton) of colleges? Will that coalition be able to manipulate the SPDA rule? Recall that for the marriage problem (a one-to-one matching market), we get a negative answer to such a question (Theorem \ref{theorem structure of manipulative coalitions}). However, it turns out that is not the case for the college admissions problem; a group of students can indeed manipulate the SPDA rule by colluding with a non-empty group of colleges. Therefore, Theorem \ref{theorem structure of manipulative coalitions} cannot be extended to the college admissions problem. We provide an example (Example \ref{example new one}) below to show it.

	\begin{example}\label{example new one}
		Consider a market with three colleges $C = \{c_1, c_2, c_3\}$ and five students $S = \{s_1, \ldots, s_5\}$. College $c_1$ has a quota of 2 and other colleges have a quota of 1. Consider the preference profile $P_I$ such that
		\begin{equation*}
			\begin{aligned}
				& P_{s_1} : c_3 c_1 \ldots, \hspace{2 mm} P_{s_2} : c_1 c_3 \ldots, \hspace{2 mm} P_{s_3} : c_1 \emptyset \ldots, \hspace{2 mm} P_{s_4} : c_1 c_2 \ldots, \hspace{2 mm} P_{s_5} : c_2 \emptyset \ldots,\\
				& P_{c_1} : \{s_1,s_2\} \hspace{1 mm} \{s_1,s_3\} \hspace{1 mm} \{s_1,s_4\} \hspace{1 mm} \{s_2,s_3\} \hspace{1 mm} \{s_2,s_4\} \hspace{1 mm} \{s_3,s_4\} \hspace{1 mm} \{s_1\} \hspace{1 mm} \{s_2\} \hspace{1 mm} \{s_3\} \hspace{1 mm} \{s_4\} \hspace{1 mm} \emptyset \hspace{1 mm} \{s_1,s_5\} \hspace{1 mm} \{s_2,s_5\} \hspace{1 mm} \{s_3,s_5\} \hspace{1 mm} \{s_4,s_5\} \hspace{1 mm} \{s_5\}, \\
				& P_{c_2} : \{s_4\} \hspace{1 mm} \{s_5\} \hspace{1 mm} \emptyset \hspace{1 mm} \{s_1\} \hspace{1 mm} \{s_2\} \hspace{1 mm} \{s_3\}, \hspace{1 mm} \mbox{ and } \hspace{1 mm} P_{c_3} : \{s_2\} \hspace{1 mm} \{s_5\} \hspace{1 mm} \{s_1\} \hspace{1 mm} \emptyset \hspace{1 mm} \{s_3\} \hspace{1 mm} \{s_4\}.
			\end{aligned}
		\end{equation*}
		Clearly, the preferences of colleges satisfy responsiveness. The outcome of the SPDA rule at $P_I$ is $$\big[(c_1, \{s_2,s_3\}), (c_2, \{s_4\}), (c_3, \{s_1\}), (s_5, \emptyset)\big].$$
		
		Let $I'$ be a coalition of student $s_5$ and college $c_1$. Consider the preference profile $\tilde{P}_{I'}$ of this coalition such that
		\begin{equation*}
			\begin{aligned}
				& \tilde{P}_{s_5} : c_3 c_2 \ldots, \hspace{1 mm} \mbox{ and}\\
				& \tilde{P}_{c_1} : \{s_1,s_4\} \hspace{1 mm} \{s_2,s_4\} \hspace{1 mm} \{s_3,s_4\} \hspace{1 mm} \{s_1,s_2\} \hspace{1 mm} \{s_1,s_3\} \hspace{1 mm} \{s_2,s_3\} \hspace{1 mm} \{s_4\} \hspace{1 mm} \{s_1\} \hspace{1 mm} \{s_2\} \hspace{1 mm} \{s_3\} \hspace{1 mm} \emptyset \hspace{1 mm} \{s_4,s_5\} \hspace{1 mm} \{s_1,s_5\} \hspace{1 mm} \{s_2,s_5\} \hspace{1 mm} \{s_3,s_5\} \hspace{1 mm} \{s_5\}.
			\end{aligned}
		\end{equation*}
		$\tilde{P}_{c_1}$ also satisfies responsiveness and the outcome of the SPDA rule at $(\tilde{P}_{I'}, P_{-I'})$ is $$\big[(c_1, \{s_1,s_4\}), (c_2, \{s_5\}), (c_3, \{s_2\}), (s_3, \emptyset)\big].$$
		
		Combining all these facts, it follows that the coalition of student $s_5$ and college $c_1$ manipulates the SPDA rule at $P_I$ via $\tilde{P}_{I'}$.
		\hfill
		$\Diamond$
	\end{example}

	Example \ref{example new one} also shows the impossibility of extending Propositions \ref{proposition update in welfare} and \ref{proposition same set of matched agents} to the college admissions problem. In fact, one important implication of Proposition \ref{proposition same set of matched agents} -- \textit{an unmatched agent cannot be a part of manipulating the DA rule} -- does not hold for the college admissions problem either. To see this, notice that in Example \ref{example new one}, student $s_5$, an unmatched agent, is a part of the manipulative coalition.
	
	Lastly, the main takeaway of this paper -- \textit{equivalence between strategy-proofness and group strategy-proofness under stability} (Theorem \ref{theorem sp wgsp equivalence}) -- is also impossible to extend to the college admissions problem. A detailed explanation is provided below. We use an additional notion for this explanation. Given a preference $P_c$ of a college $c$, let $P^*_c$ denote the corresponding \textit{induced preference} over $S \cup \{\emptyset\}$ where for every $s, s' \in S$, (i) $s \mathrel{P^*_c} s' \iff \{s\} \mathrel{P_c} \{s'\}$ and (ii) $s \mathrel{P^*_c} \emptyset \iff \{s\} \mathrel{P_c} \emptyset$.\bigskip

	\noindent \textbf{Example \ref{example new one}} \textit{(continued)}.
	Construct a domain of preference profiles $\mathcal{P}_I$ for the given market such that $\mathcal{P}_I$ satisfies unrestricted top pairs for students and 
	\begin{equation*}
		\mathcal{P}_{c_1} = \{P_{c_1}, \tilde{P}_{c_1}\}, \hspace{2 mm} \mathcal{P}_{c_2} = \{P_{c_2}\}, \hspace{1 mm} \mbox{ and } \hspace{1 mm} \mathcal{P}_{c_3} = \{P_{c_3}\}.
	\end{equation*}
	Notice that $\mathcal{P}_I$ satisfies top dominance for colleges when judged by the induced preferences. Because of this, and since colleges have responsive preferences, by Theorem 5 in \citet{alcalde1994top}, the SPDA rule is stable and strategy-proof on $\mathcal{P}_I$. 
	
	Recall that the coalition $I'$ of student $s_5$ and college $c_1$ manipulates the SPDA rule at $P_I$ via $\tilde{P}_{I'}$. Moreover, by construction, both $P_I$ and $(\tilde{P}_{I'}, P_{-I'})$ are admissible preference profiles. Combining all these facts, it follows that the SPDA rule is not group strategy-proof on $\mathcal{P}_I$, showing the impossibility of extending Proposition \ref{proposition wgsp existence} and Theorem \ref{theorem sp wgsp equivalence} to the college admissions problem.
	\hfill
	$\Diamond$

	\appendixtocon
	\appendixtitletocon

	\renewcommand{\thetable}{\Alph{section}.\arabic{table}}
	\renewcommand{\thefigure}{\Alph{section}.\arabic{figure}}	
	\renewcommand{\theequation}{\Alph{section}.\arabic{equation}}
	
	\begin{appendices}

		\section{Proof of Proposition \ref{proposition update in welfare}}\label{appendix proof of proposition update in welfare}
		
		\setcounter{equation}{0}
		\setcounter{table}{0}
		\setcounter{figure}{0}

		Before we start proving Theorem \ref{theorem structure of manipulative coalitions}, to facilitate the proof, we present a lemma that was formulated by J. S. Hwang and is proved in \citet{gale1985some}.

		\begin{lemma}[Blocking Lemma]\label{lemma hwang}
			Consider a preference profile $P_A \in \mathbb{L}^p(W \cup \{\emptyset\}) \times \mathbb{L}^q(M \cup \{\emptyset\})$. Let $\mu$ be any individually rational matching at $P_A$ and $M' := \{m \in M \mid \mu(m) \mathrel{P_m} D^M_m(P_A)\}$ the set of men who are strictly better off under $\mu$ than under $D^M(P_A)$. If $M'$ is non-empty, then there is a pair $(m, w) \in (M \setminus M') \times \mu(M')$ that blocks $\mu$ at $P_A$.
		\end{lemma}

		\begin{proof}[\textbf{Completion of the proof of Proposition \ref{proposition update in welfare}.}]
			Since $D^M$ is manipulable at $P_A$ by coalition $A'$ via $\tilde{P}_{A'}$, we have
			\begin{equation}\label{equation claim}
				D^M_a(\tilde{P}_{A'},P_{-A'}) \mathrel{P_a} D^M_a(P_A) \mbox{ for all } a \in A'.
			\end{equation}
			For ease of presentation, we denote the matching $D^M(P_A)$ by $\mu$ and the matching $D^M(\tilde{P}_{A'},P_{-A'})$ by $\tilde{\mu}$ in this proof.

			\paragraph{Proof of part \ref{item structure men weakly worse}.} We first show that $\tilde{\mu}$ is individually rational at $P_A$.
			Since $\tilde{\mu}$ is stable at $(\tilde{P}_{A'},P_{-A'})$ (see Remark \ref{remark gale result}), we have
			\begin{equation}\label{equation IR 1}
				\tilde{\mu}(a) \mathrel{R_a} \emptyset \mbox { for all } a \in A \setminus A'.
			\end{equation}
			Furthermore, $\mu$ being stable at $P_A$ (see Remark \ref{remark gale result}) implies that $\mu (a) \mathrel{R_a} \emptyset$ for all $a \in A'$. This, along with \eqref{equation claim}, yields 
			\begin{equation}\label{equation IR 2}
				\tilde{\mu}(a) \mathrel{P_a} \emptyset \mbox { for all } a \in A'.
			\end{equation} 
			\eqref{equation IR 1} and \eqref{equation IR 2} together imply that $\tilde{\mu}$ is individually rational at $P_A$.
			
			We now proceed to complete the proof of part \ref{item structure men weakly worse}. Let $M^{+} := \{m \in M \mid \tilde{\mu}(m) \mathrel{P_m} \mu(m)\}$ be the set of men who are strictly better off under $\tilde{\mu}$ than under $\mu$. Assume for contradiction that $M^{+}$ is non-empty. Since $\tilde{\mu}$ is individually rational at $P_A$ and $M^{+}$ is non-empty, by Lemma \ref{lemma hwang}, it follows that there is a pair $(m, w) \in (M \setminus M^{+}) \times \tilde{\mu}(M^{+})$ that blocks $\tilde{\mu}$ at $P_A$. 
			
			\begin{claim}\label{claim no change in preferences}
				$m, w \in A \setminus A'$.
			\end{claim}
			
			\begin{claimproof}[Proof of Claim \ref{claim no change in preferences}.]
				Note that $M \cap A' \subseteq M^{+}$. Since $m \in M \setminus M^{+}$, this implies 
				\begin{equation}\label{equation claim 1}
					m \in A \setminus A'.
				\end{equation}
				
				Consider the man $m'$ such that $\tilde{\mu}(m') = w$. Note that $m'$ is well-defined since $w \in \tilde{\mu}(M^{+})$. Clearly, $m' \in M^{+}$. The facts $m' \in M^{+}$ and $\tilde{\mu}(m') = w$ together imply 
				\begin{equation}\label{equation claim 2}
					w \mathrel{P_{m'}} \mu(m').
				\end{equation}
				Since $\mu$ is stable at $P_A$, \eqref{equation claim 2} implies $\mu(w) \mathrel{P_w} m'$. This, along with the fact $\tilde{\mu}(m') = w$ and \eqref{equation claim}, yields 
				\begin{equation}\label{equation claim 3}
					w \in A \setminus A'.
				\end{equation}
				\eqref{equation claim 1} and \eqref{equation claim 3} together complete the proof of Claim \ref{claim no change in preferences}.
			\end{claimproof}

			Recall that $\tilde{\mu}$ is stable at $(\tilde{P}_{A'},P_{-A'})$. By Claim \ref{claim no change in preferences}, it follows that $m$ and $w$ do not change their preferences from $P_A$ to $(\tilde{P}_{A'},P_{-A'})$. This, together with the fact that $(m, w)$ blocks $\tilde{\mu}$ at $P_A$, implies that $(m, w)$ blocks $\tilde{\mu}$ at $(\tilde{P}_{A'},P_{-A'})$, a contradiction to the fact that $\tilde{\mu}$ is stable at $(\tilde{P}_{A'},P_{-A'})$. This completes the proof of part \ref{item structure men weakly worse} of Theorem \ref{theorem structure of manipulative coalitions}.

			\paragraph{Proof of part \ref{item structure women weakly better}.} Assume for contradiction that there exists $w \in W$ such that
			\begin{equation}\label{equation women worse}
				\mu(w) \mathrel{P_w} \tilde{\mu}(w).
			\end{equation}
			
			We first show that $\mu(w) \in M$.
			Note that \eqref{equation women worse} and \eqref{equation claim} together imply $w \in A \setminus A'$, which means woman $w$ does not change her preference from $P_A$ to $(\tilde{P}_{A'},P_{-A'})$. Since $\tilde{\mu}$ is stable at $(\tilde{P}_{A'},P_{-A'})$ and woman $w$ does not change her preference from $P_A$ to $(\tilde{P}_{A'},P_{-A'})$, we have $\tilde{\mu}(w) \mathrel{R_w} \emptyset$, which, along with \eqref{equation women worse}, yields $\mu(w) \mathrel{P_w} \emptyset$. This, in particular, means $\mu(w) \in M$.
			
			Consider the man $m$ such that $\mu(w) = m$. By part \ref{item structure men weakly worse} of this theorem, we have $w \mathrel{R_m} \tilde{\mu}(m)$. Note also that \eqref{equation women worse} implies $\tilde{\mu}(m) \neq w$. Combining the facts $w \mathrel{R_m} \tilde{\mu}(m)$ and $\tilde{\mu}(m) \neq w$, we have
			\begin{equation}\label{equation man matched}
				w \mathrel{P_m} \tilde{\mu}(m).
			\end{equation}
			
			It follows from \eqref{equation women worse} and \eqref{equation man matched} that $(m, w)$ blocks $\tilde{\mu}$ at $P_A$. Recall that $\tilde{\mu}$ is stable at $(\tilde{P}_{A'},P_{-A'})$ and woman $w$ does not change her preference from $P_A$ to $(\tilde{P}_{A'},P_{-A'})$. Also note that by Theorem \ref{theorem structure of manipulative coalitions}, man $m$ does not change his preference from $P_A$ to $(\tilde{P}_{A'},P_{-A'})$. Combining all these facts, it follows that $(m, w)$ blocks $\tilde{\mu}$ at $(\tilde{P}_{A'},P_{-A'})$, a contradiction to the fact that $\tilde{\mu}$ is stable at $(\tilde{P}_{A'},P_{-A'})$. This completes the proof of part \ref{item structure women weakly better} of Theorem \ref{theorem structure of manipulative coalitions}.
		\end{proof}

		\section{Proof of Proposition \ref{proposition same set of matched agents}}\label{appendix proof of proposition same set of matched agents}
		
		\setcounter{equation}{0}
		\setcounter{table}{0}
		\setcounter{figure}{0}

		For ease of presentation, we denote the matching $D^M(P_A)$ by $\mu$ and the matching $D^M(\tilde{P}_{A'},P_{-A'})$ by $\tilde{\mu}$ in this proof.
		
		Note that by Theorem \ref{theorem structure of manipulative coalitions}, men do not change their preferences from $P_A$ to $(\tilde{P}_{A'},P_{-A'})$. Because of this, and since $\tilde{\mu}$ is stable at $(\tilde{P}_{A'},P_{-A'})$ (see Remark \ref{remark gale result}), we have $\tilde{\mu}(m) \mathrel{R_m} \emptyset$ for all $m \in M$. This, together with Proposition \ref{proposition update in welfare}.\ref{item structure men weakly worse}, implies
		\begin{equation*}
			\mu(m) \mathrel{R_m} \tilde{\mu}(m) \mathrel{R_m} \emptyset \mbox { for all } m \in M,
		\end{equation*}
		which, in particular, means
		\begin{equation}\label{equation men subset}
			\{m \in M \mid m \mbox{ is matched under } \tilde{\mu}\} \subseteq \{m \in M \mid m \mbox{ is matched under } \mu\}.
		\end{equation}
		
		Similarly, since $\mu$ is stable at $P_A$ (see Remark \ref{remark gale result}), we have $\mu(w) \mathrel{R_w} \emptyset$ for all $w \in W$. This, together with Proposition \ref{proposition update in welfare}.\ref{item structure women weakly better}, implies
		\begin{equation*}
			\tilde{\mu}(w) \mathrel{R_w} \mu(w) \mathrel{R_w} \emptyset \mbox { for all } w \in W,
		\end{equation*}
		which, in particular, means
		\begin{equation}\label{equation women subset}
			\{w \in W \mid w \mbox{ is matched under } \tilde{\mu}\} \supseteq \{w \in W \mid w \mbox{ is matched under } \mu\}.
		\end{equation}
		
		Furthermore, by the definition of a matching, 
		\begin{equation}\label{equation same number}
			\begin{aligned}
				& |\{m \in M \mid m \mbox{ is matched under } \tilde{\mu}\}| = |\{w \in W \mid w \mbox{ is matched under } \tilde{\mu}\}|, \mbox{ and}\\
				& |\{m \in M \mid m \mbox{ is matched under } \mu\}| = |\{w \in W \mid w \mbox{ is matched under } \mu\}|.
			\end{aligned}
		\end{equation}
		\eqref{equation men subset}, \eqref{equation women subset}, and \eqref{equation same number} together complete the proof of Proposition \ref{proposition same set of matched agents}.
		\hfill
		\qed

		\section{Proof of Theorem \ref{theorem sp wgsp equivalence}}\label{appendix proof of theorem sp wgsp equivalence}
		
		\setcounter{equation}{0}
		\setcounter{table}{0}
		\setcounter{figure}{0}

		We first prove a couple of lemmas that we use in the proof of Theorem \ref{theorem sp wgsp equivalence}.

		\subsection{Lemmas \ref{lemma only if} and \ref{lemma sufficient for incompatibility} and their proofs}

		Our first lemma (Lemma \ref{lemma only if}) says that whenever the domain satisfies unrestricted top pairs for men, if there exists a stable and strategy-proof matching rule, it must be the MPDA rule. \citet{alcalde1994top} prove a similar result (see Theorem 3 in \citet{alcalde1994top}) under the assumption that the sets of admissible preferences for men are unrestricted. Our result extends theirs.

		\begin{lemma}\label{lemma only if}
			Let $\mathcal{P}_A$ satisfy unrestricted top pairs for men. If there exists a stable and strategy-proof matching rule on $\mathcal{P}_A$, then the rule must be unique and is the MPDA rule.
		\end{lemma}

		Before we start proving Lemma \ref{lemma only if}, to facilitate the proof, we present two results -- one from \citet{gale1962college} and the other from \citet{mcvitie1970stable}.

		\begin{remark}[\citealp{gale1962college}]\label{remark men better}
			Every man weakly prefers the match by the MPDA rule to the match under any other stable matching. Formally, for every $P_A \in \mathbb{L}^p(W \cup \{\emptyset\}) \times \mathbb{L}^q(M \cup \{\emptyset\})$, every stable matching $\mu$ at $P_A$, and every $m \in M$, we have $D^M_m(P_A) \mathrel{R_m} \mu(m)$.
		\end{remark}

		\begin{remark}[\citealp{mcvitie1970stable}]\label{remark same unmatched set}
			The set of unmatched agents remains the same across the stable matchings at a preference profile. Formally, for every $P_A \in \mathbb{L}^p(W \cup \{\emptyset\}) \times \mathbb{L}^q(M \cup \{\emptyset\})$, every stable matchings $\mu$ and $\nu$ at $P_A$, and every $a \in A$, $\mu(a) = \emptyset$ if and only if $\nu(a) = \emptyset$.
		\end{remark}

		\begin{proof}[\textbf{Proof of Lemma \ref{lemma only if}}]
			Let $\varphi$ be a stable and strategy-proof matching rule on $\mathcal{P}_A$. Assume for contradiction that $\varphi \not\equiv D^M$ on $\mathcal{P}_A$. Then, there exist $P_A \in \mathcal{P}_A$ and $m \in M$ such that $\varphi_m(P_A) \neq D^M_m(P_A)$. Because of this, and since $\varphi(P_A)$ is stable at $P_A$, by Remark \ref{remark men better}, we have $D^M_m(P_A) \in W$ and
			\begin{equation}\label{equation mpda better}
				D^M_m(P_A) \mathrel{P_m} \varphi_m(P_A).
			\end{equation}
			
			Consider the preference $\tilde{P}_m \in \mathcal{P}_m$ such that $D^M_m(P_A) \mathrel{\tilde{P}_m} \emptyset \mathrel{\tilde{P}_m} w$ for all $w \in W \setminus \{D^M_m(P_A)\}$. Note that $\tilde{P}_m$ is well-defined since $\mathcal{P}_A$ satisfies unrestricted top pairs for men and $D^M_m(P_A) \in W$. Clearly, $D^M(P_A)$ is stable at $(\tilde{P}_m, P_{-m})$.
			
			However, since $D^M(P_A)$ is stable at $(\tilde{P}_m, P_{-m})$, by Remark \ref{remark same unmatched set} and the construction of $\tilde{P}_m$, we have $\varphi_m(\tilde{P}_m, P_{-m}) = D^M_m(P_A)$, which, together with \eqref{equation mpda better}, contradicts strategy-proofness of $\varphi$ on $\mathcal{P}_A$. This completes the proof of Lemma \ref{lemma only if}.
		\end{proof}

		\begin{note}\label{note}
			From an examination of the proof of Lemma \ref{lemma only if}, we can relax the richness condition ``unrestricted top pairs for men'' in Lemma \ref{lemma only if}. In fact, Condition \ref{item unrestricted top} of Definition \ref{definition unrestricted top pairs} is sufficient for the proof. 
			In view of this observation, by Proposition \ref{proposition wgsp existence}, we have the following result: 
			\textit{Suppose $\mathcal{P}_A$ satisfies Condition \ref{item unrestricted top} of Definition \ref{definition unrestricted top pairs} for men and top dominance for women. Then, the MPDA rule is the unique stable and strategy-proof matching rule on $\mathcal{P}_A$, which is also group strategy-proof.}
		\end{note}

		Our second lemma (Lemma \ref{lemma sufficient for incompatibility}) identifies a richness condition of the domain for women under which stability and strategy-proofness become incompatible whenever the domain satisfies unrestricted top pairs for men.

		\begin{lemma}\label{lemma sufficient for incompatibility}
			Let $\mathcal{P}_A$ satisfy unrestricted top pairs for men. 
			Suppose there exists an alternating sequence $m^1, w^1, m^2, w^2, \ldots, m^k, w^k$ of distinct men and women such that 
			\begin{enumerate}[(i)]
				\item for every $i = 2, \ldots, k$, there exists $P_{w^i} \in \mathcal{P}_{w^i}$ with $m^i \mathrel{P_{w^i}} m^{i-1} \mathrel{P_{w^i}} \emptyset$, and
				
				\item there exist $P_{w^1}, \tilde{P}_{w^1} \in \mathcal{P}_{w^1}$ with $m^1 \mathrel{P_{w^1}} m^{k} \mathrel{P_{w^1}} \emptyset$ such that for some $z \in M \cup \{\emptyset\}$,
				\begin{enumerate}[(a)]
					\item $m^k \mathrel{P_{w^1}} z$, and
					
					\item $z \mathrel{\tilde{R}_{w^1}} \emptyset$ and $m^1 \mathrel{\tilde{P}_{w^1}} z \mathrel{\tilde{P}_{w^1}} m^k$.
				\end{enumerate}
			\end{enumerate}
			Then, no stable matching rule on $\mathcal{P}_A$ is strategy-proof.
		\end{lemma}

		\begin{proof}[\textbf{Proof of Lemma \ref{lemma sufficient for incompatibility}.}]
			Suppose there exists an alternating sequence $m^1, w^1, m^2, w^2,$ $\ldots, m^k, w^k$ of distinct men and women such that 
			\begin{enumerate}[(i)]
				\item for every $i = 2, \ldots, k$, there exists $P_{w^i} \in \mathcal{P}_{w^i}$ with $m^i \mathrel{P_{w^i}} m^{i-1} \mathrel{P_{w^i}} \emptyset$, and
				
				\item there exist $P_{w^1}, \tilde{P}_{w^1} \in \mathcal{P}_{w^1}$ with $m^1 \mathrel{P_{w^1}} m^{k} \mathrel{P_{w^1}} \emptyset$ such that for some $z \in M \cup \{\emptyset\}$,
				\begin{enumerate}[(a)]
					\item $m^k \mathrel{P_{w^1}} z$, and
					
					\item $z \mathrel{\tilde{R}_{w^1}} \emptyset$ and $m^1 \mathrel{\tilde{P}_{w^1}} z \mathrel{\tilde{P}_{w^1}} m^k$.
				\end{enumerate}
			\end{enumerate}
			Assume for contradiction that there exists a stable and strategy-proof matching rule on $\mathcal{P}_A$. Note that since $\mathcal{P}_A$ satisfies unrestricted top pairs for men, by Lemma \ref{lemma only if}, it must be the MPDA rule. We distinguish the following two cases.\medskip
			
			\noindent\textbf{\textsc{Case} 1}: Suppose $z = \emptyset$.
			
			Since $\mathcal{P}_A$ satisfies unrestricted top pairs for men, we can construct a collection of preferences $P_{-\{w^1, \ldots, w^k\}}$ of all agents except for women $w^1, \ldots, w^k$ such that
			\begin{equation*}
				\begin{aligned}
					& P_{m^i} : w^{i+1} w^i \ldots \mbox{ for all } i = 1, \ldots, k-1, \\
					& P_{m^k} : w^1 w^k \ldots, \mbox{ and}\\
					& \tau(P_m) = \emptyset \mbox{ for all } m \notin \{m^1, \ldots, m^k\}.
				\end{aligned}
			\end{equation*}
			(Recall that $w^1 w^2 \ldots$ denotes a preference that ranks $w^1$ first and $w^2$ second.)
			
			It is straightforward to verify the following facts.
			\begin{equation}\label{equation strongest case 1}
				\begin{aligned}
					& D^M(P_{w^1}, P_{w^2}, \ldots, P_{w^k}, P_{-\{w^1, \ldots, w^k\}}) = \left[ 
					\begin{aligned}
						& (m^k,w^1), (m^{i},w^{i+1}) \hspace{2 mm} \forall \hspace{2 mm} i = 1, \ldots, k-1,\\
						& (a, \emptyset) \hspace{2 mm} \forall \hspace{2 mm} a \notin \{m^1, \ldots, m^k, w^1, \ldots, w^k\} 
					\end{aligned}
					\right] \mbox{ and}\\
					& D^M(\tilde{P}_{w^1}, P_{w^2}, \ldots, P_{w^k}, P_{-\{w^1, \ldots, w^k\}}) = \left[ 
					\begin{aligned}
						& (m^i,w^i) \hspace{2 mm} \forall \hspace{2 mm} i = 1, \ldots, k,\\
						& (a, \emptyset) \hspace{2 mm} \forall \hspace{2 mm} a \notin \{m^1, \ldots, m^k, w^1, \ldots, w^k\} 
					\end{aligned}
					\right].
				\end{aligned}
			\end{equation}
			However, \eqref{equation strongest case 1} implies that $w^1$ can manipulate the MPDA rule at $(P_{w^1}, P_{w^2}, \ldots, P_{w^k}, P_{-\{w^1, \ldots, w^k\}})$ via $\tilde{P}_{w^1}$, a contradiction to the fact that the MPDA rule is strategy-proof on $\mathcal{P}_A$. This completes the proof for Case 1.\medskip
			
			\noindent\textbf{\textsc{Case} 2}: Suppose $z = \tilde{m}$ for some $\tilde{m} \in M$.
			\begin{enumerate}[(i)]
				\item Suppose $\tilde{m} \notin \{m^1, \ldots, m^k\}$.
				
				Since $\mathcal{P}_A$ satisfies unrestricted top pairs for men, we can construct a collection of preferences $P_{-\{w^1, \ldots, w^k\}}$ of all agents except for women $w^1, \ldots, w^k$ such that
				\begin{equation*}
					\begin{aligned}
						& P_{m^i} : w^{i+1} w^i \ldots \mbox{ for all } i = 1, \ldots, k-1, \\
						& P_{m^k} : w^1 w^k \ldots,\\
						& P_{\tilde{m}} : w^1 \emptyset \ldots, \mbox{ and}\\
						& \tau(P_m) = \emptyset \mbox{ for all } m \notin \{m^1, \ldots, m^k, \tilde{m}\}.
					\end{aligned}
				\end{equation*}
				Using a similar argument as for Case 1, it follows that $w^1$ can manipulate the MPDA rule at $(P_{w^1}, P_{w^2},$ $\ldots,P_{w^k}, P_{-\{w^1, \ldots, w^k\}})$ via $\tilde{P}_{w^1}$, a contradiction to the fact that the MPDA rule is strategy-proof on $\mathcal{P}_A$.
				
				\item Suppose $\tilde{m} \in \{m^1, \ldots, m^k\}$.
				
				Since $\tilde{m} \in \{m^1, \ldots, m^k\}$, the fact $m^1 \mathrel{\tilde{P}_{w^1}} \tilde{m} \mathrel{\tilde{P}_{w^1}} m^k$ implies $\tilde{m} \in \{m^2, \ldots, m^{k-1}\}$. Let $\tilde{m} = m^{k^*}$ for some $k^* \in \{2, \ldots, k-1\}$. Since $\mathcal{P}_A$ satisfies unrestricted top pairs for men, we can construct a collection of preferences $P_{-\{w^1, \ldots, w^k\}}$ of all agents except for women $w^1, \ldots, w^k$ such that
				\begin{equation*}
					\begin{aligned}
						& P_{m^i} : w^{i+1} w^i \ldots \mbox{ for all } i = 1, \ldots, k^* - 1,\\
						& P_{m^{k^*}} : w^1 w^{k^*} \ldots,\\
						& P_{m^k} : w^1 \emptyset \ldots, \mbox{ and}\\
						& \tau(P_m) = \emptyset \mbox{ for all } m \notin \{m^1, \ldots, m^{k^*}, m^k\}.
					\end{aligned}
				\end{equation*}
				It is straightforward to verify the following facts.
				\begin{equation}\label{equation strongest case 2}
					\begin{aligned}
						& D^M(\tilde{P}_{w^1}, P_{w^2}, \ldots, P_{w^k}, P_{-\{w^1, \ldots, w^k\}}) = \left[ 
						\begin{aligned}
							& (m^{k^*},w^1), (m^{i},w^{i+1}) \hspace{2 mm} \forall \hspace{2 mm} i = 1, \ldots, k^*-1,\\
							& (a, \emptyset) \hspace{2 mm} \forall \hspace{2 mm} a \notin \{m^1, \ldots, m^{k^*}, w^1, \ldots, w^{k^*}\} 
						\end{aligned}
						\right] \mbox{ and}\\
						& D^M(P_{w^1}, P_{w^2}, \ldots, P_{w^k}, P_{-\{w^1, \ldots, w^k\}}) = \left[ 
						\begin{aligned}
							& (m^i,w^i) \hspace{2 mm} \forall \hspace{2 mm} i = 1, \ldots, k^*,\\
							& (a, \emptyset) \hspace{2 mm} \forall \hspace{2 mm} a \notin \{m^1, \ldots, m^{k^*}, w^1, \ldots, w^{k^*}\} 
						\end{aligned}
						\right].
					\end{aligned}
				\end{equation}
				However, \eqref{equation strongest case 2} implies that $w^1$ can manipulate the MPDA rule at $(\tilde{P}_{w^1}, P_{w^2}, \ldots, P_{w^k}, P_{-\{w^1, \ldots, w^k\}})$ via $P_{w^1}$, a contradiction to the fact that the MPDA rule is strategy-proof on $\mathcal{P}_A$. This completes the proof for Case 2.
			\end{enumerate}
			Since Cases 1 and 2 are exhaustive, this completes the proof of Lemma \ref{lemma sufficient for incompatibility}.
		\end{proof}

		\subsection{Completion of the proof of Theorem \ref{theorem sp wgsp equivalence}}

		The ``if'' part of the theorem follows from the respective definitions. We proceed to prove the ``only-if'' part. Without loss of generality, assume that $\mathcal{P}_A$ satisfies unrestricted top pairs for men. Suppose there exists a stable and strategy-proof matching rule on $\mathcal{P}_A$. By Lemma \ref{lemma only if}, it must be the MPDA rule.
		Assume for contradiction that the MPDA rule $D^M$ is not group strategy-proof on $\mathcal{P}_A$. Then, there exist $P_A \in \mathcal{P}_A$, a set of agents $A' \subseteq A$, and a preference profile $\tilde{P}_{A'}$ of the agents in $A'$ such that
		\begin{equation}\label{equation assumption 1}
			D^M_a(\tilde{P}_{A'},P_{-A'}) \mathrel{P_a} D^M_a(P_A) \mbox{ for all } a \in A'.
		\end{equation}
		For ease of presentation, we denote the matching $D^M(P_A)$ by $\mu$ and the matching $D^M(\tilde{P}_{A'},P_{-A'})$ by $\tilde{\mu}$ in this proof. Furthermore, let $\tilde{\mu}^s$ denote the (tentative) matching at the end of some step $s$ of the MPDA rule at $(\tilde{P}_{A'},P_{-A'})$.
		
		Note that by Theorem \ref{theorem structure of manipulative coalitions} and Proposition \ref{proposition update in welfare}.\ref{item structure men weakly worse}, men do not change their preferences from $P_A$ to $(\tilde{P}_{A'},P_{-A'})$ and each man weakly prefers $\mu$ to $\tilde{\mu}$ where $\mu \neq \tilde{\mu}$. Let $s^*$ be first step (of the MPDA rule) at $(\tilde{P}_{A'},P_{-A'})$ when some man, say $m$, gets rejected by $\mu(m)$. Consider the woman $w^1$ such that $\mu(m) = w^1$. Clearly, $\tilde{\mu}(w^1) \neq m$. Construct an alternating sequence $m^1, w^1, m^2, w^2, \ldots, m^k, w^k$ of distinct men and women with $m^k \equiv m$ such that
		\begin{enumerate}[(i)]
			\item $\tilde{\mu}(w^i) = m^i$ for all $i = 1, \ldots, k$, and
			
			\item $\mu(m^i) = w^{i+1}$ for all $i = 1, \ldots, k-1$.
		\end{enumerate}
		Since both the number of men and the number of women are finite, by Proposition \ref{proposition same set of matched agents}, it follows that the constructed sequence is well-defined.
		
		Since $\mu$ is stable at $P_A$ (see Remark \ref{remark gale result}), it follows from the construction of the sequence and Proposition \ref{proposition update in welfare}.\ref{item structure women weakly better} that
		\begin{subequations}\label{equation women preference}
			\begin{equation}\label{equation women preference 1}
				m^1 \mathrel{P_{w^1}} m^{k} \mathrel{P_{w^1}} \emptyset, \mbox{ and}
			\end{equation}
			\begin{equation}\label{equation women preference 2}
				m^i \mathrel{P_{w^i}} m^{i-1} \mathrel{P_{w^i}} \emptyset \mbox{ for all } i = 2, \ldots, k.
			\end{equation}
		\end{subequations}
		We distinguish the following two cases.\medskip
		
		\noindent\textbf{\textsc{Case} 1}: Suppose $\tilde{\mu}^{s^*}(w^1) = \emptyset$.
		
		Since $m^k$ gets rejected by $w^1$ in Step $s^*$ (of the MPDA rule) at $(\tilde{P}_{A'},P_{-A'})$, the fact $\tilde{\mu}^{s^*}(w^1) = \emptyset$, together with \eqref{equation women preference 1}, implies $w^1 \in A'$ and
		\begin{equation}\label{equation women updated preference 1}
			\emptyset \mathrel{\tilde{P}_{w^1}} m^k.
		\end{equation}
		Moreover, since $\tilde{\mu}$ is stable at $(\tilde{P}_{A'},P_{-A'})$ (see Remark \ref{remark gale result}), $\tilde{\mu}(w^1) = m^1$ implies $m^1 \mathrel{\tilde{P}_{w^1}} \emptyset$, which, along with \eqref{equation women updated preference 1}, yields
		\begin{equation}\label{equation women updated preference}
			m^1 \mathrel{\tilde{P}_{w^1}} \emptyset \mathrel{\tilde{P}_{w^1}} m^k.
		\end{equation}
		
		However, by Lemma \ref{lemma sufficient for incompatibility}, \eqref{equation women preference} and \eqref{equation women updated preference} together contradict the fact that the MPDA rule is both stable and strategy-proof on $\mathcal{P}_A$. This completes the proof for Case 1.\medskip
		
		\noindent\textbf{\textsc{Case} 2}: Suppose $\tilde{\mu}^{s^*}(w^1) \in M$.
		
		Consider the man $\tilde{m}$ such that $\tilde{\mu}^{s^*}(w^1) = \tilde{m}$. Since $m^k$ gets rejected by $w^1$ in Step $s^*$ (of the MPDA rule) at $(\tilde{P}_{A'},P_{-A'})$, the fact $\tilde{\mu}^{s^*}(w^1) = \tilde{m}$ implies $\tilde{m} \neq m^k$. Moreover, since man $\tilde{m}$ does not change his preference from $P_A$ to $(\tilde{P}_{A'},P_{-A'})$, by the assumption of Step $s^*$ being the first step at $(\tilde{P}_{A'},P_{-A'})$ where some man gets rejected by his match under $\mu$, $\tilde{\mu}^{s^*}(w^1) = \tilde{m}$ implies $w^1 \mathrel{R_{\tilde{m}}} \mu(\tilde{m})$. This, along with the facts $\mu(m^k) = w^1$ and $\tilde{m} \neq m^k$, yields $w^1 \mathrel{P_{\tilde{m}}} \mu(\tilde{m})$. Because of this, and since $\mu$ is stable at $P_A$ with $\mu(m^k) = w^1$, we have
		\begin{equation}\label{equation w1 pref}
			m^k \mathrel{P_{w^1}} \tilde{m}.
		\end{equation}
		Furthermore, since $m^k$ gets rejected by $w^1$ in Step $s^*$ at $(\tilde{P}_{A'},P_{-A'})$, the fact $\tilde{\mu}^{s^*}(w^1) = \tilde{m}$, together with \eqref{equation w1 pref} implies $w^1 \in A'$ and
		\begin{subequations}\label{equation woman up preference}
			\begin{equation}\label{equation woman up preference 1}
				\tilde{m} \mathrel{\tilde{P}_{w^1}} \emptyset, \mbox{ and}
			\end{equation}
			\begin{equation}\label{equation woman up preference 2}
				\tilde{m} \mathrel{\tilde{P}_{w^1}} m^k.
			\end{equation}
		\end{subequations}
		Note that \eqref{equation women preference 1} and \eqref{equation w1 pref} together imply $\tilde{m} \neq m^1$.
		
		By the definition of the MPDA rule, we have $\tilde{\mu}(w^1) \mathrel{\tilde{R}_{w^1}} \tilde{\mu}^{s^*}(w^1)$. Because of this, and since $\tilde{\mu}(w^1) = m^1$, $\tilde{\mu}^{s^*}(w^1) = \tilde{m}$, and $\tilde{m} \neq m^1$, we have $m^1 \mathrel{\tilde{P}_{w^1}} \tilde{m}$, which, along with \eqref{equation woman up preference 2}, yields
		\begin{equation}\label{equation woman up preference 3}
			m^1 \mathrel{\tilde{P}_{w^1}} \tilde{m} \mathrel{\tilde{P}_{w^1}} m^k.
		\end{equation}
		
		However, by Lemma \ref{lemma sufficient for incompatibility}, \eqref{equation women preference}, \eqref{equation w1 pref}, \eqref{equation woman up preference 1}, and \eqref{equation woman up preference 3} together contradict the fact that the MPDA rule is both stable and strategy-proof on $\mathcal{P}_A$. This completes the proof for Case 2.\medskip
		
		Since Cases 1 and 2 are exhaustive, this completes the proof of Theorem \ref{theorem sp wgsp equivalence}.
		\hfill
		\qed

		\section{Proof of Proposition \ref{proposition wgsp existence}}\label{appendix proof of proposition wgsp existence}
		
		\setcounter{equation}{0}
		\setcounter{table}{0}
		\setcounter{figure}{0}

		We use Theorem \ref{theorem sp wgsp equivalence} (which is presented after Proposition \ref{proposition wgsp existence} in the body of the paper) in the proof of Proposition \ref{proposition wgsp existence}. Because of this, we have already presented the proof of Theorem \ref{theorem sp wgsp equivalence} in Appendix \ref{appendix proof of theorem sp wgsp equivalence}.

		\begin{proof}[\textbf{Completion of the proof of Proposition \ref{proposition wgsp existence}}]
			Construct a domain of preference profiles $\mathcal{\tilde{P}}_A$ such that $\mathcal{\tilde{P}}_m = \mathbb{L}(W \cup \{\emptyset\})$ for all $m \in M$ and $\mathcal{\tilde{P}}_w = \mathcal{P}_w$ for all $w \in W$. Clearly, $\mathcal{\tilde{P}}_A$ satisfies unrestricted top pairs for men and top dominance for women. 
			
			Since $\mathcal{\tilde{P}}_A$ satisfies top dominance for women, by the corollary in \citet{alcalde1994top}, the MPDA rule is stable and strategy-proof on $\mathcal{\tilde{P}}_A$. Because of this, and since $\mathcal{\tilde{P}}_A$ satisfies unrestricted top pairs for men, by Theorem \ref{theorem sp wgsp equivalence}, it follows that the MPDA rule is also group strategy-proof on $\mathcal{\tilde{P}}_A$.
			
			However, since the MPDA rule is stable and group strategy-proof on $\mathcal{\tilde{P}}_A$, it also satisfies stability and group strategy-proofness on the smaller domain $\mathcal{P}_A$. This completes the proof of Proposition \ref{proposition wgsp existence}.
		\end{proof}

		\section{Proof of Theorem \ref{theorem single-peaked}}\label{appendix proof of theorem single-peaked}
		
		\setcounter{equation}{0}
		\setcounter{table}{0}
		\setcounter{figure}{0}

		We demonstrate the desired implications for the equivalence in turn.\medskip
		
		\noindent {\textbf{\ref{item TD} $\implies$ \ref{item DA is gsp}.}} 
		This implication follows from Proposition \ref{proposition wgsp existence}.
		
		\noindent {\textbf{\ref{item DA is gsp} $\implies$ \ref{item stable and gsp} $\implies$ \ref{item stable and sp}.}} 
		These implications are straightforward.\medskip
		
		Notice that the above implications are independent of single-peakedness, anonymity, or cyclical inclusion.\medskip
		
		\noindent {\textbf{\ref{item stable and sp} $\implies$ \ref{item TD}.}} 
		We prove the contrapositive: Suppose $\mathcal{P}_A$ satisfies top dominance for neither of the sides, then no stable matching rule on $\mathcal{P}_A$ is strategy-proof.
		Since $\mathcal{P}_A$ is an anonymous domain, for ease of presentation, we denote the common sets of admissible preferences for men and women by $\mathcal{P}_{men}$ and $\mathcal{P}_{women}$, respectively.
		
		Since $\mathcal{P}_A$ does not satisfy top dominance for men, there exist $P_1, P_2 \in \mathcal{P}_{men}$ such that for some $x \in W$ and some $y, z \in W \cup \{\emptyset\}$,
		\begin{equation*}
			\begin{aligned}
				& x \mathrel{P_1} y \mathrel{P_1} z \hspace{1 mm} \mbox{ and } \hspace{1 mm} y \mathrel{R_1} \emptyset, \mbox{ and}\\
				& x \mathrel{P_2} z \mathrel{P_2} y \hspace{1 mm} \mbox{ and } \hspace{1 mm} z \mathrel{R_2} \emptyset.
			\end{aligned}
		\end{equation*} 
		We can distinguish two options. One, where either $y$ or $z$ is the outside option $\emptyset$. The other, where both $y$ and $z$ are women. Without loss of generality, let these two options be:
		\begin{equation*}
			\begin{aligned}
				& \mbox{Option 1 (for men)}: x = w_2, y = w_1, \mbox{ and } z = \emptyset \\
				& \mbox{Option 2 (for men)}: x = w_2, y = w_1, \mbox{ and } z = w_3
			\end{aligned}
		\end{equation*}
		Notice that for both options, we have
		\begin{equation}\label{equation assumption for both cases men}
			w_2 \mathrel{P_1} w_1 \mathrel{P_1} \emptyset.
		\end{equation}
		
		Similarly, since $\mathcal{P}_A$ does not satisfy top dominance for women, there exist $\tilde{P}_1, \tilde{P}_2 \in \mathcal{P}_{women}$ such that for some $\tilde{x} \in M$ and some $\tilde{y}, \tilde{z} \in M \cup \{\emptyset\}$,
		\begin{equation*}
			\begin{aligned}
				& \tilde{x} \mathrel{\tilde{P}_1} \tilde{y} \mathrel{\tilde{P}_1} \tilde{z} \hspace{1 mm} \mbox{ and } \hspace{1 mm} \tilde{y} \mathrel{\tilde{R}_1} \emptyset, \mbox{ and}\\
				& \tilde{x} \mathrel{\tilde{P}_2} \tilde{z} \mathrel{\tilde{P}_2} \tilde{y} \hspace{1 mm} \mbox{ and } \hspace{1 mm} \tilde{z} \mathrel{\tilde{R}_2} \emptyset.
			\end{aligned}
		\end{equation*}
		We can also distinguish two options for women. One, where either $\tilde{y}$ or $\tilde{z}$ is the outside option $\emptyset$. The other, where both $\tilde{y}$ and $\tilde{z}$ are men. Without loss of generality, let these two options be:
		\begin{equation*}
			\begin{aligned}
				& \mbox{Option 1 (for women)}: \tilde{x} = m_2, \tilde{y} = m_1, \mbox{ and } \tilde{z} = \emptyset \\
				& \mbox{Option 2 (for women)}: \tilde{x} = m_2, \tilde{y} = m_1, \mbox{ and } \tilde{z} = m_3
			\end{aligned}
		\end{equation*}
		Notice that for both options, we have
		\begin{equation}\label{equation assumption for both cases women}
			m_2 \mathrel{\tilde{P}_1} m_1 \mathrel{\tilde{P}_1} \emptyset.
		\end{equation}
		
		Furthermore, since $\mathcal{P}_A$ is an anonymous domain that satisfies cyclical inclusion for both sides, \eqref{equation assumption for both cases men} implies there exists a preference $P_3 \in \mathcal{P}_{men}$ such that 
		\begin{equation}\label{equation assumption for both cases men 1}
			w_1 \mathrel{P_3} w_2 \mathrel{P_3} \emptyset,
		\end{equation}
		and \eqref{equation assumption for both cases women} implies there exists a preference $\tilde{P}_3 \in \mathcal{P}_{women}$ such that 
		\begin{equation}\label{equation assumption for both cases women 1}
			m_1 \mathrel{\tilde{P}_3} m_2 \mathrel{\tilde{P}_3} \emptyset.
		\end{equation}
		We distinguish the following four cases.\medskip
		
		\noindent\textbf{\textsc{Case} 1}: Suppose Option 1 holds for both men and women.
		
		Since $\mathcal{P}_A$ is an anonymous domain that satisfies cyclical inclusion for both sides, we can construct the preference profiles presented in Table \ref{preference profiles one side has to be TD case 1}. Here, $m_k$ denotes a man other than $m_1, m_2$ (if any), and $w_l$ denotes a woman other than $w_1, w_2$ (if any). Note that such an agent does not change his/her preference across the mentioned preference profiles.
		\begin{table}[H]
			\centering
			\begin{tabular}{@{}c|ccc|ccc@{}}
				\hline
				Preference profiles & $m_1$ & $m_2$ & $m_k$ & $w_1$ & $w_2$ & $w_l$ \\ \hline
				\hline
				$P^1_A$ & $P_3$ & $P_1$ & $\emptyset \ldots$ & $\tilde{P}_1$ & $\tilde{P}_3$ & $\emptyset \ldots$ \\
				$P^2_A$ & $P_3$ & $P_2$ & $\emptyset \ldots$ & $\tilde{P}_1$ & $\tilde{P}_3$ & $\emptyset \ldots$ \\
				$P^3_A$ & $P_3$ & $P_1$ & $\emptyset \ldots$ & $\tilde{P}_2$ & $\tilde{P}_3$ & $\emptyset \ldots$ \\
				\hline
			\end{tabular}
			\caption{Preference profiles for Case 1}
			\label{preference profiles one side has to be TD case 1}
		\end{table}
		
		For ease of presentation, let $\mu$ denote the matching $\big[(m_1,w_1), (m_2,w_2), (a,\emptyset) \hspace{2 mm} \forall \hspace{2 mm} a \in A \setminus \{m_1, m_2, w_1, w_2\}\big]$ and $\tilde{\mu}$ the matching $\big[(m_1,w_2), (m_2,w_1), (a,\emptyset) \hspace{2 mm} \forall \hspace{2 mm} a \in A \setminus \{m_1, m_2, w_1, w_2\}\big]$ for this case. 
		The sets of stable matchings at $P^1_A$, $P^2_A$, and $P^1_A$ are $\{\mu, \tilde{\mu}\}$, $\{\mu\}$, and $\{\tilde{\mu}\}$, respectively. 
		
		Fix a stable matching rule $\varphi$ on $\mathcal{P}_A$. If $\varphi(P^1_A) = \mu$, then $w_1$ can manipulate at $P^1_A$ via $\tilde{P}_2$. If $\varphi(P^1_A) = \tilde{\mu}$, then $m_2$ can manipulate at $P^1_A$ via $P_2$. This implies $\varphi$ is not strategy-proof on $\mathcal{P}_A$.\medskip
		
		\noindent\textbf{\textsc{Case} 2}: Suppose Option 1 holds for men and Option 2 holds for women.
		
		Since $\mathcal{P}_A$ is a single-peaked domain, the facts $m_2 \mathrel{\tilde{P}_1} m_1 \mathrel{\tilde{P}_1} m_3$ and $m_2 \mathrel{\tilde{P}_2} m_3 \mathrel{\tilde{P}_2} m_1$ together imply
		\begin{equation*}
			m_1 \prec_M m_2 \prec_M m_3 \hspace{1 mm} \mbox{ or } \hspace{1 mm} m_3 \prec_M m_2 \prec_M m_1.
		\end{equation*}
		This, together with \eqref{equation assumption for both cases women 1}, implies
		\begin{equation*}
			m_1 \mathrel{\tilde{P}_3} m_2 \mathrel{\tilde{P}_3} m_3.
		\end{equation*}
		
		Moreover, since $\mathcal{P}_A$ is an anonymous domain that satisfies cyclical inclusion for both sides, we can construct the preference profiles presented in Table \ref{preference profiles one side has to be TD case 2}.
		\begin{table}[H]
			\centering
			\begin{tabular}{@{}c|cccc|ccc@{}}
				\hline
				Preference profiles & $m_1$ & $m_2$ & $m_3$ & $m_k$ & $w_1$ & $w_2$ & $w_l$ \\ \hline
				\hline
				$P^1_A$ & $P_3$ & $P_1$ & $P_3$ & $\emptyset \ldots$ & $\tilde{P}_1$ & $\tilde{P}_3$ & $\emptyset \ldots$ \\
				$P^2_A$ & $P_3$ & $P_2$ & $P_3$ & $\emptyset \ldots$ & $\tilde{P}_1$ & $\tilde{P}_3$ & $\emptyset \ldots$ \\
				$P^3_A$ & $P_3$ & $P_1$ & $P_3$ & $\emptyset \ldots$ & $\tilde{P}_2$ & $\tilde{P}_3$ & $\emptyset \ldots$ \\
				\hline
			\end{tabular}
			\caption{Preference profiles for Case 2}
			\label{preference profiles one side has to be TD case 2}
		\end{table}
		
		Using a similar argument as for Case 1, it follows that no stable matching rule on $\mathcal{P}_A$ is strategy-proof.\medskip
		
		\noindent\textbf{\textsc{Case} 3}: Suppose Option 2 holds for men and Option 1 holds for women.
		
		Since $\mathcal{P}_A$ is a single-peaked domain, the facts $w_2 \mathrel{P_1} w_1 \mathrel{P_1} w_3$ and $w_2 \mathrel{P_2} w_3 \mathrel{P_2} w_1$ together imply
		\begin{equation*}
			w_1 \prec_W w_2 \prec_W w_3 \hspace{1 mm} \mbox{ or } \hspace{1 mm} w_3 \prec_W w_2 \prec_W w_1.
		\end{equation*}
		This, together with \eqref{equation assumption for both cases men 1}, implies
		\begin{equation*}
			w_1 \mathrel{P_3} w_2 \mathrel{P_3} w_3.
		\end{equation*}
		
		Moreover, since $\mathcal{P}_A$ is an anonymous domain that satisfies cyclical inclusion for both sides, we can construct the preference profiles presented in Table \ref{preference profiles one side has to be TD case 3}.
		\begin{table}[H]
			\centering
			\begin{tabular}{@{}c|ccc|cccc@{}}
				\hline
				Preference profiles & $m_1$ & $m_2$ & $m_k$ & $w_1$ & $w_2$ & $w_3$ & $w_l$ \\ \hline
				\hline
				$P^1_A$ & $P_3$ & $P_1$ & $\emptyset \ldots$ & $\tilde{P}_1$ & $\tilde{P}_3$ & $\tilde{P}_3$ & $\emptyset \ldots$ \\
				$P^2_A$ & $P_3$ & $P_2$ & $\emptyset \ldots$ & $\tilde{P}_1$ & $\tilde{P}_3$ & $\tilde{P}_3$ & $\emptyset \ldots$ \\
				$P^3_A$ & $P_3$ & $P_1$ & $\emptyset \ldots$ & $\tilde{P}_2$ & $\tilde{P}_3$ & $\tilde{P}_3$ & $\emptyset \ldots$ \\
				\hline
			\end{tabular}
			\caption{Preference profiles for Case 3}
			\label{preference profiles one side has to be TD case 3}
		\end{table}
		
		Using a similar argument as for Case 1, it follows that no stable matching rule on $\mathcal{P}_A$ is strategy-proof.\medskip
		
		\noindent\textbf{\textsc{Case} 4}: Suppose Option 2 holds for both men and women.
		
		Since $\mathcal{P}_A$ is a single-peaked domain, the facts $w_2 \mathrel{P_1} w_1 \mathrel{P_1} w_3$ and $w_2 \mathrel{P_2} w_3 \mathrel{P_2} w_1$ together imply
		\begin{equation*}
			w_1 \prec_W w_2 \prec_W w_3 \hspace{1 mm} \mbox{ or } \hspace{1 mm} w_3 \prec_W w_2 \prec_W w_1.
		\end{equation*}
		This, together with \eqref{equation assumption for both cases men 1}, implies
		\begin{equation*}
			w_1 \mathrel{P_3} w_2 \mathrel{P_3} w_3.
		\end{equation*}
		Using a similar argument as above, we also have 
		\begin{equation*}
			m_1 \mathrel{\tilde{P}_3} m_2 \mathrel{\tilde{P}_3} m_3.
		\end{equation*}
		
		Moreover, since $\mathcal{P}_A$ is an anonymous domain that satisfies cyclical inclusion for both sides, we can construct the preference profiles presented in Table \ref{preference profiles one side has to be TD case 4}.
		\begin{table}[H]
			\centering
			\begin{tabular}{@{}c|cccc|cccc@{}}
				\hline
				Preference profiles & $m_1$ & $m_2$ & $m_3$ & $m_k$ & $w_1$ & $w_2$ & $w_3$ & $w_l$ \\ \hline
				\hline
				$P^1_A$ & $P_3$ & $P_1$ & $P_3$ & $\emptyset \ldots$ & $\tilde{P}_1$ & $\tilde{P}_3$ & $\tilde{P}_3$ & $\emptyset \ldots$ \\
				$P^2_A$ & $P_3$ & $P_2$ & $P_3$ & $\emptyset \ldots$ & $\tilde{P}_1$ & $\tilde{P}_3$ & $\tilde{P}_3$ & $\emptyset \ldots$ \\
				$P^3_A$ & $P_3$ & $P_1$ & $P_3$ & $\emptyset \ldots$ & $\tilde{P}_2$ & $\tilde{P}_3$ & $\tilde{P}_3$ & $\emptyset \ldots$ \\
				\hline
			\end{tabular}
			\caption{Preference profiles for Case 4}
			\label{preference profiles one side has to be TD case 4}
		\end{table}
		
		\begin{enumerate}[(i)]
			\item Suppose $w_3 \mathrel{P_3} \emptyset$ and $m_3 \mathrel{\tilde{P}_3} \emptyset$.
			
			Let $\mu$ denote the matching $\big[(m_1,w_1), (m_2,w_2), (m_3,w_3), (a,\emptyset) \hspace{2 mm} \forall \hspace{2 mm} a \in A \setminus \{m_1, m_2, m_3, w_1, w_2, w_3\}\big]$ and $\tilde{\mu}$ the matching $\big[(m_1,w_2), (m_2,w_1), (m_3,w_3), (a,\emptyset) \hspace{2 mm} \forall \hspace{2 mm} a \in A \setminus \{m_1, m_2, m_3, w_1, w_2, w_3\}\big]$. 
			The sets of stable matchings at $P^1_A$, $P^2_A$, and $P^1_A$ are $\{\mu, \tilde{\mu}\}$, $\{\mu\}$, and $\{\tilde{\mu}\}$, respectively. 
			
			Fix a stable matching rule $\varphi$ on $\mathcal{P}_A$. If $\varphi(P^1_A) = \mu$, then $w_1$ can manipulate at $P^1_A$ via $\tilde{P}_2$. If $\varphi(P^1_A) = \tilde{\mu}$, then $m_2$ can manipulate at $P^1_A$ via $P_2$. This implies $\varphi$ is not strategy-proof on $\mathcal{P}_A$.
			
			\item Suppose either $\emptyset \mathrel{P_3} w_3$ or $\emptyset \mathrel{\tilde{P}_3} m_3$ or both.
			
			Using a similar argument as for Case 1, it follows that no stable matching rule on $\mathcal{P}_A$ is strategy-proof.
		\end{enumerate}
		Since Cases 1 -- 4 are exhaustive, this completes the proof of Theorem \ref{theorem single-peaked}.
		\hfill
		\qed

	\end{appendices}

	\bibliographystyle{plainnat}
	\setcitestyle{numbers}
	\bibliography{mybib}

\end{document}